\documentclass{sig-alternate}

\newfont{\mycrnotice}{ptmr8t at 7pt}
\newfont{\myconfname}{ptmri8t at 7pt}
\let\crnotice\mycrnotice%
\let\confname\myconfname%

\permission{Permission to make digital or hard copies of part or all of this work for personal or classroom use is granted without fee provided that copies are not made or distributed for profit or commercial advantage, and that copies bear this notice and the full citation on the first page. Copyrights for third-party components of this work must be honored. For all other uses, contact the owner/author(s). Copyright is held by the author/owner(s).}

\conferenceinfo{ISSAC'15,}{July 6--9, 2015, Bath, United Kingdom.}
\copyrightetc{ACM \the\acmcopyr}
\crdata{978-1-4503-3435-8/15/07. \\
DOI: http://dx.doi.org/10.1145/2755996.2756678} 

\usepackage{url} 
\usepackage{bm} 
\usepackage{enumerate} 
\usepackage{lhelp} 

\usepackage{algorithmic,multirow}

\usepackage[ruled,vlined,linesnumbered]{algorithm2e}

\newtheorem{theorem}{Theorem}

\newtheorem{lemma}[theorem]{Lemma}
\newtheorem{definition}{Definition}
\newtheorem{example}{Example}

\def\cont{\mathop{\rm cont}\nolimits}
\def\prim{\mathop{\rm prim}\nolimits}
\def\disc{\mathop{\rm disc}\nolimits}
\def\ldcf{\mathop{\rm ldcf}\nolimits}
\def\coeff{\mathop{\rm coeff}\nolimits}
\def\res{\mathop{\rm res}\nolimits}

\begin{document}

\title{Improving the use of equational constraints in \\ cylindrical algebraic decomposition}

\numberofauthors{3}
\author{
\alignauthor
Matthew England\\
       \affaddr{Coventry University}\\
       \email{\normalsize{\texttt{Matthew.England@coventry.ac.uk}}}
\alignauthor
Russell Bradford\\
       \affaddr{University of Bath}\\
       \email{\normalsize{\texttt{R.J.Bradford@bath.ac.uk}}}
\alignauthor
James H. Davenport\\
       \affaddr{University of Bath}\\
       \email{\normalsize{\texttt{J.H.Davenport@bath.ac.uk}}}
}

\maketitle

\begin{abstract} 
When building a cylindrical algebraic decomposition (CAD) savings can be made in the presence of an equational constraint (EC): an equation logically implied by a formula.  

The present paper is concerned with how to use multiple ECs, propagating those in the input throughout the projection set.  We improve on the approach of McCallum in ISSAC 2001 by using the reduced projection theory to make savings in the lifting phase (both to the polynomials we lift with and the cells lifted over).  
We demonstrate the benefits with worked examples and a complexity analysis.
\end{abstract}


\section{Introduction}
\label{SEC:Intro}

A \emph{cylindrical algebraic decomposition} (CAD) splits $\mathbb{R}^n$ into cells arranged \emph{cylindrically}, meaning the projections of any pair are either equal or disjoint, and such that each can be described with a finite sequence of polynomial constraints.  

Introduced by Collins for quantifier elimination in real closed fields, applications of CAD include:
derivation of optimal numerical schemes \cite{EH14},
parametric optimisation \cite{FPM05}, 
epidemic modelling \cite{BENW06}, 
theorem proving \cite{Paulson2012}, 
reasoning with multi-valued functions \cite{DBEW12}, and much more.

CAD has complexity doubly exponential in the number of variables \cite{DH88}.  For some applications there exist algorithms with better complexity (see \cite{BPR06}), but CAD implementations remain the best general purpose approach for many.  This may be due to the many extensions and optimisations of CAD since Collins including: partial CAD (to lift only when necessary for quantifier elimination); symbolic-numeric lifting schemes \cite{Strzebonski2006, IYAY09}; local projection approaches \cite{Brown2013, Strzebonski2014a}; and decompositions via complex space \cite{CMXY09, BCDEMW14}.  Collins original algorithm is described in \cite{ACM84I} while a more detailed summary of recent developments can be found, for example, in \cite{BDEMW15}.

\subsection{CAD computation and terminology}

We describe the computation scheme and terminology that most CAD algorithms share.  We assume a set of input polynomials (possibly derived from formulae) in ordered variables $\bm{x} = x_1\prec \ldots \prec x_n$.  
The \emph{main variable} of a polynomial (${\rm mvar}$) is the greatest variable present under the ordering.  

The first phase of CAD, {\em projection}, applies projection operators repeatedly, each time producing another set of polynomials in one fewer variables.  Together these contain the {\em projection polynomials} used in the second phase, {\em lifting}, to build the CAD incrementally.  First $\mathbb{R}$ is decomposed into cells which are points and intervals according to the real roots of polynomials univariate in $x_1$.  Then $\mathbb{R}^2$ is decomposed by repeating the process over each cell with the bivariate polynomials in ($x_1,x_2)$ evaluated at a sample point.  

This produces {\em sections} (where a polynomial vanishes) and {\em sectors} (the regions between) which together form the {\em stack} over the cell.  Taking the union of these stacks gives the CAD of $\mathbb{R}^2$ and this is repeated until a CAD of $\mathbb{R}^n$ is produced.  

At each stage cells are represented by (at least) a sample point and an {\em index}.  The latter is a list of integers, with the $k$th describing variable $x_k$ according to the ordered real roots of the projection polynomials in $(x_1, \dots, x_k)$.  If the integer is $2i$ the cell is over the $i$th root (counting from low to high) and if $2i+1$ over the interval between the $i$th and $(i+1)$th (or the unbounded intervals at either end). 

The projection operator is chosen so polynomials are {\em delineable} in a cell: the portion of their zero set in the cell consists of disjoint sections.  A set of polynomials are {\em delineable} if each is individually, and the sections of different polynomials  are identical or disjoint.  If all projection polynomials are delineable then the input polynomials must be {\em sign-invariant}: have constant sign in each cell of the CAD.  

\subsection{Equational constraints}

Most applications of CAD require {\em truth-invariance} for logical formulae, meaning each formula has constant boolean truth value on each cell.  Sign-invariance for the polynomials in a formula gives truth-invariance, but we can obtain the latter more efficiently by using equational constraints.

\begin{definition}
We use {\em QFF} to denote a quantifier free Tarski formula: Boolean combinations ($\land,\lor,\neg$) of statements about the signs ($=0,>0,<0$) of integral polynomials.  

An {\em equational constraint} (EC) is a polynomial equation logically implied by a QFF.  If an atom of the formula it is said to be {\em explicit} and is otherwise {\em implicit}.
\end{definition}

Collins first suggested that the projection phase of CAD could be simplified in the presence of an EC \cite{Collins1998}.  He noted that a CAD sign-invariant for the defining polynomial of an EC, and sign-invariant for any others only on sections of that polynomial, would be sufficient.  An intuitive approach to produce this is to consider resultants of the EC polynomial with the other polynomials, in place of them.  This approach was first formalised and verified in \cite{McCallum1999b}.

A recent complexity analysis \cite{BDEMW15} showed that using an EC in this way reduces the double exponent in the complexity bound for CAD by $1$.  A natural question is whether this can be repeated in the presence of multiple ECs.   An algorithm for CAD in the presence of two ECs was detailed in  \cite{McCallum2001}.  The main idea was to observe that the resultant of the polynomials defining two ECs is itself an EC, and so the same ideas could be applied for the second projection as for the first.  However, this approach was complicated as the key result verifying \cite{McCallum1999b} could not be applied recursively.

\subsection{Contribution and plan}

This paper discusses how we can extend the theory of ECs to produce CADs more efficiently.  In Section \ref{SUBSEC:McC} we revise key components of the theory for reduced projection in the presence of an EC from \cite{McCallum1999b, McCallum2001}.  Then in Section \ref{SUBSEC:Lifting} we explain how it can also give reductions in the lifting phase, allowing us to propose and verify a new algorithm in Section \ref{SEC:Alg} for making use of multiple ECs.  
This breaks with the tradition of producing CADs sign-invariant for EC polynomials, instead guaranteeing only invariance for the truth of their conjunction.  
We demonstrate our contributions in Sections \ref{SEC:Example} and \ref{SEC:Complexity} with a worked example and complexity analysis.

All experiments in \textsc{Maple} were conducted using \textsc{Maple} 18.  All code and  data created for this paper is openly available from 
\url{http://dx.doi.org/10.15125/BATH-00071}.

\section{CAD with multiple equational \\ constraints}
\label{SEC:EC}

\subsection{Key theory from \cite{McCallum1998, McCallum1999b, McCallum2001}}
\label{SUBSEC:McC}

We recall some of the key theory behind McCallum's operators.  
Let $\cont$, $\prim$, $\disc$, $\coeff$ and $\ldcf$ denote the content, primitive part, discriminant, coefficients and leading coefficient of polynomials respectively (in each case taken with respect to a given mvar).  Let $\res$ denote the resultant of a pair of polynomials.  When applied to a set of polynomials we interpret these as producing sets of polynomials, e.g.
\[
\res(A)=\left\{\res(f_i,f_j) \, | \, f_i \in A, f_j \in A, f_j \neq f_i \right\}.
\]
Recall that a set $A \subset \mathbb{Z}[{\bf x}]$ is an \emph{irreducible basis} if the elements of $A$ are of positive degree in the mvar, irreducible and pairwise relatively prime.  Throughout this section suppose $B$ is an irreducible basis for a set of polynomials, that every element of $B$ has mvar $x_n$ and that $F \subseteq B$.
Define
\begin{align}
P(B) &:= \coeff(B) \cup \disc(B) \cup \res(B),
\label{eq:P}
\\
P_{F}(B) &:= P(F) \cup \{ {\rm res}(f,g) \mid f \in F, g \in B \setminus F \},
\label{eq:ECProj} 
\end{align}
\begin{align}
P_{F}^{*}(B) &:= P_{F}(B) \cup \res(B \setminus F),
\label{eq:ECProjStar}
\end{align}
as the projection operators introduced respectively in \cite{McCallum1998, McCallum1999b, McCallum2001}.  In the general case with $A$ a set of polynomials and $E \subseteq A$ we proceed with projection by: letting $B$ and $F$ be irreducible basis of the primitive parts of $A$ and $E$ respectively; applying the operators as defined above; and then taking the union of the output with $\cont(A)$.  

The theorems in this section validate the use of these operators for CAD.  They use the condition of \emph{order-invariance}, meaning each polynomial has constant order of vanishing within each cell, which of course implies sign-invariance.  We say that a polynomial with mvar $x_k$ is \emph{nullified} over a cell in $\mathbb{R}^{k-1}$ if it vanishes identically throughout.

\begin{theorem}[\cite{McCallum1998}]
\label{thm:McC1}
Let $S$ be a connected submanifold of $\mathbb{R}^{n-1}$ in which each element of $P(B)$ is order-invariant. 

Then on $S$, each element of $B$ is either nullified or analytic delineable (a variant on delineability, see \cite{McCallum1998}). Further, the sections of $B$ not nullified are pairwise disjoint, and each element of such $B$ is order-invariant on such sections.
\end{theorem}
Suppose we apply $P$ repeatedly to generate projection polynomials.  Repeated use of Theorem \ref{thm:McC1} concludes that a CAD produced by lifting with respect to these projection polynomial is order-invariant so long as no projection polynomial with mvar $x_k$ is nullified over a cell in the CAD of $\mathbb{R}^{k-1}$ (a condition known as \emph{well-orientedness} which can be checked during lifting).  If this condition is not satisfied then $P$ cannot be used (and we should restart the CAD construction using a different projection operator, such as Hong's \cite{Hong1990}).

\begin{theorem}[\cite{McCallum1999b}]
\label{thm:McC2}
Let $f$ and $g$ be integral polynomials with mvar $x_n$, $r(x_1,\ldots,x_{n-1})$ be their resultant, and suppose $r \neq 0$.
Let $S$ be a connected subset of $\mathbb{R}^{n-1}$ on which $f$ is delineable and $r$ order-invariant.

Then $g$ is sign-invariant in every section of $f$ over $S$.
\end{theorem}
Suppose $A$ was derived from a formula with EC defined by $E = \{f\}$, and that we apply $P_E(A)$ once and then $P$ repeatedly to generate a set of projection polynomials.  
Assuming the input is well-oriented, we can use Theorem \ref{thm:McC1} to conclude the CAD of $\mathbb{R}^{n-1}$ order invariant for $P_E(A)$.  The CAD of $\mathbb{R}^n$ is then sign-invariant for $E$ using Theorem \ref{thm:McC1} and sign-invariant for $A$ in the sections of $E$ using Theorem \ref{thm:McC2}.  Hence the CAD is truth-invariant for the formula. 

\vspace*{0.1in}

What if there are multiple ECs?  We could designate one for special use and treat the rest as any other constraint (heuristics can help with the choice \cite{BDEW13}).   But this does not gain any more advantage than one EC gives.  However, we cannot simply add multiple polynomials into $E$ at the top level as this would result in a CAD truth-invariant for the disjunction of the ECs, not the conjunction.

Suppose we have a formula with a second EC. If this has a lower mvar then we may consider applying the reduced projection operator again at this lower level.  In fact, even if the second EC is also in the mvar of the system we can \emph{propagate} it to the lower level by noting that the resultant of the two ECs is itself an EC in one fewer variable.  

So we consider applying first the operator $P_E(A)$ where $E$ defines the first EC and then $P_{E'}(A')$ where $A' = P_{E}(A)$ and $E' \subseteq A'$ contains the EC in one variable fewer.  Unfortunately, Theorem \ref{thm:McC2} does not validate this approach.  While it could be applied once for the CAD of $\mathbb{R}^{n-1}$ it cannot then validate the CAD of $\mathbb{R}^{n}$ because the first application of the theorem provided sign-invariance while the second requires the stronger condition of order invariance.  Note however, that this approach is acceptable if $n=3$ (since in two variables the conditions are equivalent for squarefree bases).  

\begin{figure}[t]
\caption{The polynomials from Example \ref{ex:Simple}.}
\label{fig:SimpleEx1}
\includegraphics[scale=0.28]{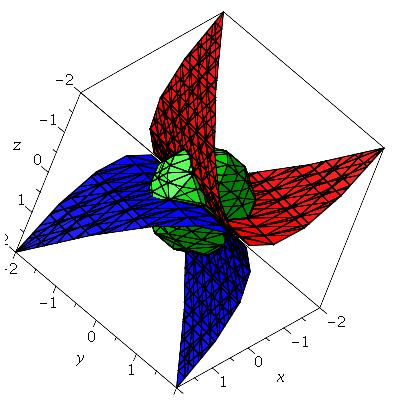}
\includegraphics[scale=0.23]{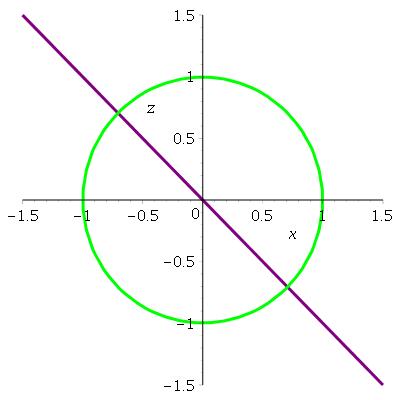}
\vskip-10pt
\end{figure}

\begin{example}
\label{ex:Simple}
The following are graphed in Figure \ref{fig:SimpleEx1}, with $g$ the sphere, $f_1$ the upper surface and $f_2$ the lower:
\[
f_1 = x+y^2+z, \quad f_2 = x-y^2+z, \quad g = x^2+y^2+z^2-1.
\]
We consider the formula $\phi = f_1=0 \land f_2=0 \land g \geq 0$.  
The surfaces $f_1$ and $f_2$ only meet on the plane $y=0$ and this projection is on the right of Figure \ref{fig:SimpleEx1}).  From this it is clear the solution requires $|x|\geq\sqrt{2}/2$ and $z=-x$.

How could this be ascertained using CAD?  With variable ordering $z \succ y \succ x$ a sign-invariant CAD for $(f_1, f_2, g)$ has 1487 cells using \textsc{Qepcad} \cite{Brown2003b}.
We could then test a sample point of each cell to identify the ones where $\phi$ is true.  

It is preferable to use the presence of ECs.  Declaring an EC to \textsc{Qepcad} will ensure it uses the algorithm in \cite{McCallum1999b} based on a single use of $P_E(A)$ followed by $P$.  Either choice results in 289 cells.  In particular, the solution set is described using 8 cells: all have $y=0, z=-x$ but the $x$-coordinate unnecessarily splits cells at $\frac{1}{2}(1\pm\sqrt{6})$.  This is identified due to the projection polynomial $d = \disc_y(\res_z(f_i, g))$.

If we declare both ECs to \textsc{Qepcad} then it will use the algorithm in \cite{McCallum2001} applying $P_E(A)$ twice (allowed since $n=3$) to produce a CAD with 133 cells.  The solution set is now described using only 4 cells (the minimum possible).  Note that $d$ was no longer produced as a projection polynomial.
\end{example}
 
For problems with $n>3$ it is still possible to make use of multiple ECs.  However, we must include the extra information necessary to provide order-invariance of the non-EC polynomials in the sections of ECs.  The following theorem may be used to conclude that $P_{E}^{*}(A)$ is appropriate.
 
\begin{theorem}[\cite{McCallum2001}]
\label{thm:McC3}
Let $f$ and $g$ be integral polynomials with mvar $x_n$, $r = \res(f,g)$,  $d=disc(g)$, and suppose $r,d \neq 0$.
Let $S$ be a connected subset of $\mathbb{R}^{n-1}$ on which $f$ is analytic delineable, $g$ is not nullified and $r$ and $d$ are order-invariant.

Then $g$ is order-invariant in each section of $f$ over $S$.
\end{theorem}

Suppose we have a formula with two ECs, one with mvar $x_n$ and the other with mvar $x_{n-1}$.  The second could be explicit in the formula or implicit (a resultant as described earlier).  Theorem \ref{thm:McC3} allows us to use a reduced operator twice.  We first calculate $A' = P_E(A)$ where $E$ contains the defining polynomial of the first EC, and then $P_{E'}^{*}(A')$ where $E'$ contains the defining polynomial of the other.  Subsequent projections simply use $P$.  When lifting we use Theorem \ref{thm:McC1} to verify the CAD of $\mathbb{R}^{n-2}$ as order-invariant for $P_{E'}^{*}(A')$; Theorem \ref{thm:McC1} to verify the CAD of $\mathbb{R}^{n-1}$ order-invariant for $E'$ everywhere and Theorem \ref{thm:McC3} to verify it order-invariant for $A'$ in the sections of $E'$; and Theorem \ref{thm:McC1} and \ref{thm:McC2} to verify the CAD of $\mathbb{R}^n$ order-invariant for $E$ and sign-invariant for $A$ in those cells that are both sections of $E$ and $E'$.

\subsection{Reductions in the lifting phase}
\label{SUBSEC:Lifting}

The main contribution of the present paper is to realise that the theorems above also allow for significant savings in the lifting phase of CAD.  However, to implement these we must discard two embedded principles of CAD:
\begin{enumerate}
\item That the projection polynomials are a fixed set.
\item That the invariance structure of the final CAD can be expressed in terms of sign-invariance of polynomials.
\end{enumerate}
Abandoning the first is key to recent work in \cite{CMXY09, BCDEMW14}, while the second was also investigated in \cite{BM05, MB09}.

\subsubsection{Minimising the polynomials when lifting}
\label{SUBSEC:IL-Poly}

Consider Theorem \ref{thm:McC2}: it allows us to conclude that $g$ is sign-invariant in the sections of $f$ produced over a CAD of $\mathbb{R}^{n-1}$ order-invariant for $P_{\{f\}}(\{f,g\})$.  Therefore, it is sufficient to perform the final lift with respect to $f$ only (decompose cylinders according to the roots of $f$ but not $g$).  The decomposition imposes sign-invariance for $f$ while Theorem \ref{thm:McC2} guarantees it for $g$ in the cells where it matters.  

\begin{example}
\label{ex:Simple2}
We return to Example \ref{ex:Simple}.  Recall that designating either EC and using \cite{McCallum1999b} produced a CAD with 289 cells.  If we follow this approach but lift only with respect to the designated EC at the final step (implemented in our \textsc{Maple} package \cite{EWBD14}) we obtain a CAD with 141 cells.  
\end{example}

This improved lifting follows from the theorems in \cite{McCallum1999b}, but was only noticed 15 years later during the generalisation of \cite{McCallum1999b} to the case of multiple formulae in \cite{BDEMW13, BDEMW15}.  Experiments there demonstrated its importance, particularly for problems with many constraints (see Section 8.3 of \cite{BDEMW15}).  

When we apply a reduced operator at two levels then we can make such reductions at both the corresponding lifts.

\begin{example}
\label{ex:Simple3}
We return to the problem from Example \ref{ex:Simple}.  Set $A = \{f_1, f_2, g\}$ and $E=\{f_1\}$.  Then project out $z$ using
\[
P_E(A) = \{ y^2, y^4+2xy^2+2x^2+y^2-1 \}.
\]
These are the resultants of $f_1$ with $f_2$ and $g$.  The discriminant of $f_1$ was a constant and so could be discarded, as was its leading coefficient (meaning no further coefficients were required).
We set $A' = P_E(A)$, $E' = \res_z(f_1,f_2) = y^2$ and 
\[
R = \res_y(y^2, y^4+2xy^2+2x^2+y^2-1) = (2x^2-1)^2.
\]
We have $P_{E'}(A') = \{ R \}$ since the other possible entries (the discriminants and coefficients from $E'$) are all constants.  We hence build a 5 cell CAD of the real line with respect to the two real roots of $R$.  We then lift above each cell with respect to $y^2$ only, in each case splitting the cylinder into three cells about $y=0$, to give a CAD of $\mathbb{R}^2$ with 15 cells.  

Finally, we lift over each of these 15 cells with respect to $f_1$ to give 45 cells of $\mathbb{R}^3$.  This compares to 133 from \textsc{Qepcad}, which used reduced projection but then lifted with all projection polynomials.  
No polynomials were nullified, so using Theorems \ref{thm:McC1} and \ref{thm:McC2}, the output is truth-invariant for $\phi$.
\end{example}

The additional lifting that \textsc{Qepcad} performed does not provide any further structure.  For example, if we had lifted with respect to $f_2$ at the final stage in Example \ref{ex:Simple3} then we would be doing so without the knowledge that it is delineable.  Hence splitting the cylinder at the sample point offers no guarantee that the cells produced are sign-invariant away from that point.  So the extra work does not allow us to conclude that $f_2$ is sign-invariant (except on sections of $f_1$).  

Note that using fewer projection polynomials for lifting not only decreases output size (and computation time) but also the risk of failure from non well-oriented input: we only need worry about nullification of polynomials we lift with.

\subsubsection{Minimising the cells for stack generation}
\label{SUBSEC:IL-Stack}

We can achieve still more savings from the theory in Section \ref{SUBSEC:McC} by abandoning the aim of producing a CAD sign-invariant with respect to any polynomial, instead insisting only on truth-invariance for the formula.  We may then lift trivially to cylinders over cells already known to be false, only identifying sections of projection polynomials if there is a possibility the formula may be true.  The idea of avoiding computations over false cells was presented in \cite{Seidl2006}.  Our contribution is to explain how such cells can easily be identified in the presence of ECs.
We demonstrate with our example.

\begin{example}
\label{ex:Simple4}
Return to the problem from Examples~\ref{ex:Simple}~$-$~\ref{ex:Simple3} and in particular the CAD of $\mathbb{R}^2$ produced with 15 cells in Example \ref{ex:Simple3}.  On 5 of these 15 cells the polynomial $R$ is zero and on the others it is either positive or negative throughout.  

Now, $\phi$ can only be satisfied above the 5 cells, as elsewhere the two EC defining polynomials cannot share a root and thus vanish together.  We can already conclude the truth value for the 10 cells (false) and thus we do not need to lift over them, except in the trivial sense of extending them to a cylinder in $\mathbb{R}^3$.  Lifting over the 5 cells where $R=0$ with respect to $f_1$ gives 15 cells, which combined with the 10 cylinders gives a CAD of $\mathbb{R}^3$ with 25 cells that is truth-invariant for $\phi$.
\end{example}

This 25 cell CAD is not sign-invariant for $f_1$.  The cylinders above the 10 cells in $\mathbb{R}^2$ where $R \neq 0$ may have $f_1$ varying sign, but since $f_1$ can never equal zero at the same time as $f_2$ in these cells it does not affect the truth of $\phi$.  

Identifying the 5 cells where $R=0$ in the CAD of $\mathbb{R}^2$ was trivial since they are simply the sections of the second lift, and hence those cells with second entry even in the cell index.  Those sections produced in the third lift are similarly all cells where $f_1$ is zero, however, we cannot conclude that $f_2$ is also zero on these.  Theorem \ref{thm:McC2} only guarantees that $f_2$ is sign-invariant on such cells, so to determine those signs we must still evaluate the polynomials at the sample point. 

Reducing the number of cells for stack generation clearly decreases output size, and since the cells can be identified using only a parity check on an integer, computation time decreases also.  As with the improvements in Section \ref{SUBSEC:IL-Poly}, this also decreases the risk of non well-oriented input: we only need worry about nullification over these identified cells.

\section{Algorithm}
\label{SEC:Alg}

We present Algorithm \ref{alg:ECM} to build a truth-invariant CAD for a formula in the presence of multiple ECs.  We assume that the ECs are already identified as input to the algorithm (they may have been first computed through propagation as described in Section \ref{SEC:EC}).  We assume further that each EC is primitive, and that all the ECs have different mvar (so in practice a choice of designation may have been made).   

\begin{algorithm}[t!]\label{alg:ECM}
\caption{CAD using multiple ECs}
\SetKwInOut{Input}{Input}\SetKwInOut{Output}{Output}
\Input{A formula $\phi$ in variables $x_1,\ldots,x_n$, and a sequence of sets $\{E_k\}_{k=1}^n$; each either empty or containing a single primitive polynomial with mvar $x_k$ which defines an EC for $\phi$.
}
\Output{Either: $\mathcal{D}$, a truth-invariant CAD of $\mathbb{R}^n$ for $\phi$ (described by lists $I$ and $S$
   of cell indices and sample points); or 
{\bf FAIL}, if not well-oriented.
}
\BlankLine
Extract from $\phi$ the set of defining polynomials $A_n$\label{step:Pstart}\;
\For{$k = n, \dots, 2$}{
  Set $B_k$ to the finest squarefree basis for ${\rm prim}(A_k)$\;
  Set $C$ to $\cont(A_k)$\;
  Set $F_k$ to the finest squarefree basis for $E_k$\;
  \eIf{$F_k$ is empty}{
    Set $A_{k-1} := C \cup P(B_k)$\label{step:P}\;
  }{
    \eIf{$k=n$ or $k=2$}{Set $A_{k-1} := C \cup P_{F_i}(B_i)$\;\label{step:PEA}}
    {Set $A_{k-1} := C \cup P_{F_i}^{*}(B_i)$\;\label{step:Pend}}
  }
}
If $E_1$ is not empty then set $p$ to be its element; otherwise set $p$ to the product of polynomials in $A_1$\label{step:Bstart}\;  
Build $\mathcal{D}_1 := (I_1,S_1)$ according to the real roots of~$p$\;
\If{$n=1$}{
\Return $\mathcal{D}_1$\;\label{step:Bend}
}
\For{$k=2, \dots, n$\label{step:Lstart}}{
  Initialise $\mathcal{D}_k = (I_k, S_k)$ with $I_k$ and $S_k$ empty sets\;
  \eIf{$F_{k}$ is empty\label{step:L0}}{
    Set $L:=B_k$\;\label{step:L1}
  }{
  Set $L:=F_k$\;\label{step:L2}
  } 
  \eIf{$E_{k-1}$ is empty\label{step:C0}}{
    Set $\mathcal{C}_a := \mathcal{D}_{k-1}$ and $\mathcal{C}_b$ empty\label{step:C1a}\;
  }{
    Set $\mathcal{C}_a$ to be cells in $\mathcal{D}_{k-1}$ with $I_{k-1}[-1]$ even\label{step:C1b}\;
    Set $\mathcal{C}_b := \mathcal{D}_{k-1} \setminus \mathcal{C}_a$\label{step:C2}\;
  }
  \For{each cell $c \in \mathcal{C}_a$}{
    \If{An element of $L$ is nullified over $c$\label{step:WO}}
    {\Return FAIL\;\label{step:fail}
    }
    Generate a stack over $c$ with respect to the polynomials in $L$, adding cell indices and sample points to $I_k$ and $S_k$\label{step:lift}\;
  }
  \For{each cell $c \in \mathcal{C}_b$}{
    Extend to a single cell in $\mathbb{R}^k$ (cylinder over $c$), adding index and sample point to $I_k$ and $S_k$\label{step:Lend}\;
  }
}
\Return $\mathcal{D}_n = (I_n, S_n)$.
\end{algorithm}

Steps \ref{step:Pstart} $-$ \ref{step:Pend} run the projection phase of the algorithm.  
Each projection starts by identifying contents and primitive parts.  
When there is no declared EC ($E_i$ is empty) the projection operator (\ref{eq:P}) is used (step \ref{step:P}).  Otherwise the operator (\ref{eq:ECProjStar}) is used (step \ref{step:Pend}), unless it is the very first or very last projection (step \ref{step:PEA}) when we use (\ref{eq:ECProj}).  In each case the output of the projection operator is combined with the contents to form the next layer of projection polynomials.  

Steps \ref{step:Bstart} $-$ \ref{step:Bend} construct a CAD for the real line (and return it if the input was univariate).  This is sometimes referred to in the literature as the \emph{base phase}.  If there is a declared EC in the smallest variable then the real line is decomposed according to its roots, otherwise according to the roots of all the univariate projection polynomials.

Steps \ref{step:Lstart} $-$ \ref{step:Lend} run the lifting phase, incrementally building CADs of $\mathbb{R}^k$ for $k=2, \dots, n$.  For each $k$ there are two considerations.  First, whether there is a declared EC with mvar $x_k$.  If so we lift only with respect to this (step \ref{step:L2}) and if not we use all projection polynomials with mvar $x_k$ (step \ref{step:L1}).  
Second, whether there is a declared EC with mvar $x_{k-1}$.  If so we restrict stack generation to those cells where the EC was satisfied.  These are simply those with $I_{k-1}[-1]$ (last entry in the cell index) even (step \ref{step:C1b}).  We lift the other cells trivially to a cylinder in step \ref{step:Lend}.

Algorithm \ref{alg:ECM} clearly terminates.  We will verify that it produces a truth-invariant CAD for the formula so long as the input is well-oriented, as defined below.

\begin{definition}
\label{def:WO}
For $k=2, \dots, n$ define sets:
\begin{itemizeshort}
\item $L_k$ $-$ the \emph{lifting polynomials}: the defining polynomial of the declared EC with mvar $x_k$ if one exists, or all projection polynomials with mvar $x_k$ otherwise.
\item $\mathcal{C}_k$ $-$ the \emph{lifting cells}: those cells in the CAD of $\mathbb{R}^{k-1}$ in which the designated EC with mvar $x_{k-1}$ vanishes if it exists, and all cells in that CAD otherwise. 
\end{itemizeshort}
The input of Algorithm \ref{alg:ECM} is \emph{well-oriented} if for $k=2, \dots, n$ no element of $L_k$ is nullified over an element of $\mathcal{C}_k$.
\end{definition}

\begin{theorem}
\label{thm:Alg}
Algorithm \ref{alg:ECM} satisfies its specification.
\end{theorem}
\begin{proof}
We must show the CAD is truth-invariant for $\phi$, unless the input is not well-oriented when FAIL is returned.

First consider the case where $n=1$.  The projection phase would not run, with the algorithm jumping to the CAD construction in step \ref{step:Bstart}, returning the output in step \ref{step:Bend}.  If there was no declared EC then the CAD is sign-invariant for all polynomials defining $\phi$ and thus every cell is truth invariant for $\phi$.  If there was a declared EC then the output is sign-invariant for its defining polynomial.  Cells would either be intervals where the formula must be false; or points, where the EC is satisfied, and the formula either identically true or false depending on the signs of the other polynomials.

Next suppose that the input were not well-oriented (Definition \ref{def:WO}).  For a fixed $k$, the conditional in steps \ref{step:L0} $-$ \ref{step:L2} sets the lifting polynomials $L_k$ to $L$ and the conditional in steps \ref{step:C0} $-$ \ref{step:C2} the lifting cells $\mathcal{C}_{k}$ to $\mathcal{C}_a$.  Thus it is exactly the conditions of Definition \ref{def:WO} which are checked by step \ref{step:WO}, returning FAIL in step \ref{step:fail} when they are not satisfied.  If the lifting phase completes then the input is well-oriented.

From now on we suppose $n>1$ and the input is well-oriented.  
For a fixed $k$ define \emph{admissible} cells to be those in the induced CAD of $\mathbb{R}^{k-1}$ where all declared ECs with mvar smaller than $x_k$ are satisfied, or to be all cells in that induced CAD if there are no such ECs.  Then let $I(k)$ be the following statement for the CADs produced by Algorithm \ref{alg:ECM}. \emph{Over admissible cells (in $\mathbb{R}^{k-1}$) the CAD of $\mathbb{R}^k$ is:
\begin{enumerate}[(a)]
\item order-invariant for any EC with mvar $x_k$; \vspace*{-0.1in}
\item order- (sign- if $k=n$) invariant  for all projection polynomials with mvar $x_k$ on sections of the EC over admissible cells, or over all admissible cells if no EC exists.
\end{enumerate}
}
We have already proved \noindent $I(1)$, and $I(n)$ may be proved by induction.  To assert the truth of $I(k)$ we note the following:
\begin{itemizeshort}
\item When $E_k$ is empty we use Theorem \ref{thm:McC1} to assert all projection polynomials with mvar $x_k$ are order-invariant in the stacks over admissible cells giving (a) and (b).  
\item When $E_k$ is not empty and $k=2$ we used the projection operator (\ref{eq:ECProj}).  Theorem \ref{thm:McC2} allows us to conclude (b) and that the EC is sign-invariant in admissible cells.  The stronger property of order-invariance follows automatically since the lifting polynomials form a squarefree basis in two variables.
\item  When $E_k$ is not empty and $k=n$ we used the projection operator (\ref{eq:ECProj}). Theorem \ref{thm:McC2} allows us to conclude (b), but also (a) since in the case $k=n$ the statement requires only sign-invariance.
\item When $E_k$ is not empty and $2<k<n$ we used the projection operator (\ref{eq:ECProjStar}).   Theorem \ref{thm:McC3} then allows us to conclude the statement.
\end{itemizeshort}
In each case the assumptions of the theorems are met by the inductive hypothesis, exactly over admissible cells as defined according to whether $E_{k-1}$ was empty or not.

From the definition of admissible cells, we know that $\phi$ is false (and thus trivially truth invariant) upon all cells in the CAD of $\mathbb{R}^n$ built over an inadmissible cell of $\mathbb{R}^k$, $k<n$.  Coupled with the truth of (a) for $k=1, \dots, n$, this implies the CAD of $\mathbb{R}^n$ is truth-invariant for the conjunction of ECs (although it may not be truth-invariant for any one individually).  The truth of (b) implies that on those cells where all ECs are satisfied, the other polynomials in $\phi$ are sign-invariant and thus $\phi$ is truth-invariant.
\end{proof}

\section{Worked example}
\label{SEC:Example}

Assume variable ordering $z \succ y \succ x \succ u \succ v$ and define  
\begin{align*}
&f_1 := x-y+z^2, \quad f_2 := z^2-u^2+v^2-1, \quad g := x^2-1, \\
&f_3 := x+y+z^2, \quad f_4 := z^2+u^2-v^2-1, \quad h := z.
\end{align*}
We consider the formula
\begin{align*}
\phi = f_1=0 \land f_2=0
\land f_3=0 \land f_4=0 \land g \geq 0 \land h \geq 0.
\end{align*}
The solution can be found manually by decomposing the system into blocks.  The surfaces $f_1$ and $f_3$ are graphed in $(x,y,z)$-space on the left of Figure \ref{fig:WE1}.  They meet only on the plane $y=0$ and this projection is shown on the right.  The surfaces $f_2$ and $f_4$ are graphed in $(z,u,v)$-space on the left of Figure \ref{fig:WE2} and meet only when $z=\pm 1$.  We consider only $z=+1$ due to $h \geq 0$, with this projection plotted on the right.  We thus see that the solution set is given by 
\[
\{ u=\pm v, x = -1, y = 0, z = 1\}.
\] 

To ascertain this by Algorithm \ref{alg:ECM} we must first propagate and designate ECs.  We choose to use $f_1$ first, calculate 
\[
\res_z(f_1, f_2) = (-u^2+v^2-x+y-1)^2
\]
and assign $r_1$ to be the square root: the defining polynomial for an EC with mvar $y$.   Similarly consider
\begin{align*}
&\res_y \big( r_1, \res_z(f_1,f_3) \big) = 16(u^2-v^2+x+1)^4, \\
&\res_y \big( r_1, \res_z(f_1,f_4) \big) = 4(u^2-v^2)^2
\end{align*}
and assign $r_2 := u^2-v^2+x+1$, $r_3 := u^2-v^2$: defining polynomials for ECs with mvar $x$ and $u$ respectively.  There is no series of resultants that leads to an EC with mvar $u$.   
We hence have
$
\{E_j\}_{k=1}^n := \{f_1\}, \{r_1\}, \{r_2\}, \{r_3\}, \{ \, \}
$
as input for Algorithm \ref{alg:ECM}, along with $\phi$.

\begin{figure}[t]
\caption{The polynomials $f_1$ and $f_3$ from Section \ref{SEC:Example}.}
\label{fig:WE1}
\includegraphics[scale=0.28]{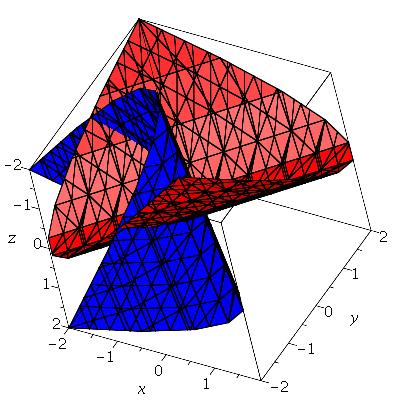}
\includegraphics[scale=0.23]{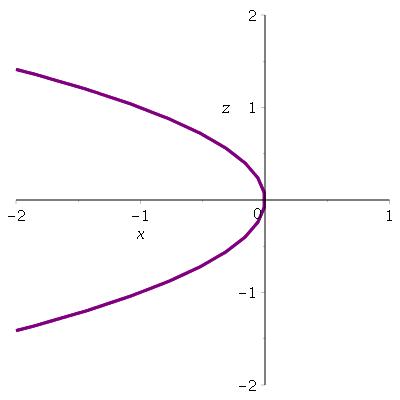}
\vskip-10pt
\end{figure}

\begin{figure}[t]
\caption{The polynomials $f_2$ and $f_4$ from Section \ref{SEC:Example}.}
\label{fig:WE2}
\includegraphics[scale=0.28]{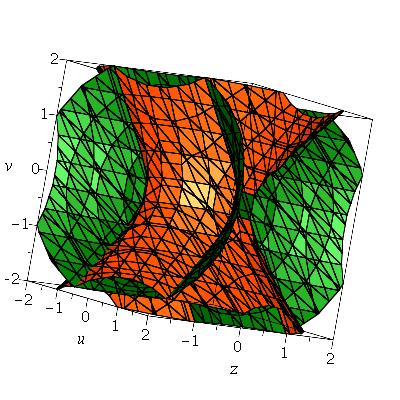}
\includegraphics[scale=0.22]{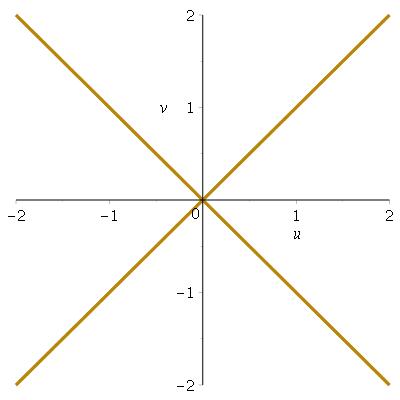}
\vskip-10pt
\end{figure}

The algorithm starts by extracting the defining polynomials $A_5 = \{f_1, f_2, f_3, f_4, g, h\}$ and finds $B_5=A_5, F_5=E_5$ (in fact $F_i=E_i$ for all $i=1, \dots, 5$).  There is a declared EC for the first projection so we use the operator (\ref{eq:ECProj}) to derive
\begin{align*}
A_4 := P_{F_5}(B_5) &= \{(x^2-1)^2, (-u^2+v^2-x+y-1)^2, \\
&\quad (u^2-v^2-x+y-1)^2, 4y^2, x-y\}.
\end{align*}
Hence $C := \{x^2-1\}$ and
\begin{align*}
B_4 &:= \{ y, y-x, -u^2+v^2-x+y-1, u^2-v^2-x+y-1 \}.
\end{align*}
For the next projection we must use operator (\ref{eq:ECProjStar}), giving
\[
A_3 := C \cup P_{F_{4}}^*(B_4) = \{x^2-1, u^2-v^2+x+1, u^2-v^2, u^2-v^2+1\}
\]
noting that for this example the extra discriminants in (\ref{eq:ECProjStar}) all evaluated to constants and so could be discarded.  Then
\[
B_3 :=  \{x^2-1, u^2-v^2+x+1\}, \quad C := \{u^2-v^2, u^2-v^2+1\},
\]
and the next projection also uses (\ref{eq:ECProjStar}) to produce
\[
A_2 := \{u^2-v^2, u^2-v^2+1, u^4-2u^2v^2+v^4+2u^2-2v^2\}.
\]
For the final projection there is no EC and so we use operator (\ref{eq:P}) to find
$A_1 := \{v^2\}.$ 
The base phase of the algorithm hence produces a 3-cell CAD of the real line isolating $0$.  

For the first lift we have $L = \{u^2-v^2\}$ and $C_a$ containing all 3 cells.  Above the two intervals we split into 5 cells by the curves $u=\pm v$, while above $v=0$ we split into three cells about the origin.
From these 13 cells of $\mathbb{R}^2$ we select the 5 which were sections of $u^2-v^2$ for $C_a$.  These are lifted with respect to $L = \{r_2\}$, and the other 8 are simply extended to cylinders in $\mathbb{R}^3$.  Together this gives a CAD of $\mathbb{R}^3$ with 23 cells.  The next two lifts are similar, producing first a CAD of $\mathbb{R}^4$ with 53 cells and finally a CAD of $\mathbb{R}^5$ with 113 cells.  The entire calculation takes less than a second in \textsc{Maple}.

\subsubsection*{Choice in EC designation}

Algorithm \ref{alg:ECM} could have been initialised with alternative EC designations.  There were the 4 explicit ECs with mvar $z$, and by taking repeated resultants we discover the following implicit ECs, organised in sets with decreasing mvar:
\begin{align*}
&\{y^2, u^2-v^2+x-y+1, -u^2+v^2+x-y+1, 
\\
&\hspace*{0.25in} u^2-v^2+x+y+1, -u^2+v^2+x+y+1 \},
\\
&\{x+1, -u^2+v^2+x+1, u^2-v^2+x+1\}, \quad 
\{u^2-v^2\}.
\end{align*}
There are hence 60 possible permutations of EC designation, but they lead to only 3 different outputs, with  113, 103 and 93 cells.  
Heuristics for other questions of CAD problem formulation \cite{DSS04, BDEW13, HEWDPB14, WEBD14} could likely be adapted to assist here.  We note that 93 cells is not a minimal truth invariant CAD for $\phi$ as it splits the CAD of $\mathbb{R}^1$ at $v=0$ (identified from the discriminant of the only EC with mvar $u$).  

\subsubsection*{Comparison with other CAD implementations}

A sign-invariant CAD of $\mathbb{R}^5$ for the 6 polynomials in the example could be produced by \textsc{Qepcad} with 1,118,205 cells.  Neither the \texttt{RegularChains} Library in \textsc{Maple} \cite{CMXY09} nor our \textsc{Maple} package \cite{EWBD14} could produce one in under an hour.  

Our implementation of \cite{McCallum1999b}, which uses operator (\ref{eq:ECProj}) once but also performs the final lift with respect to the EC only, can produce a CAD with either 3023, 10935 or 48299 (twice) cells depending on which EC is designated.  The \textsc{Qepcad} implementation of \cite{McCallum1999b} gives 11961, 30233, 158475, or 158451 cells.  Comparing these sets of figures we see the dramatic improvements from just a single reduced lift.

Allowing \textsc{Qepcad} to propagate the 4 ECs (so a similar projection phase as Algorithm \ref{alg:ECM} but then a normal CAD lifting phase) produces a CAD with 21079 cells.  By declaring only a subset of the 4 (which presumably changes the designations of implicit ECs) a CAD with 5633 cells can be produced, still much more than using Algorithm \ref{alg:ECM}.

The \texttt{RegularChains} Library can also make use of multiple ECs, as detailed in \cite{BCDEMW14}.  The version in \textsc{Maple} 18 times out after an hour, however, with the development version a CAD can be produced instantly.  There are choices (with analogies to designation \cite{EBCDMW14}) but they all lead to a 137 cell output.  In particular, they all have an induced CAD of the real line which splits at $v = \pm 1$ as well as $v=0$.

\section{Complexity analysis}
\label{SEC:Complexity}

We build on recent work in \cite{BDEMW15} to measure the dominant term in bounds on the number of CAD cells produced.  Numerous studies have shown this to be closely correlated to the computation time.  We assume input with $m$ polynomials of maximum degree $d$ in any one of $n$ variables.

\begin{definition}
\label{def:md}
Consider a set of polynomials $p_j$.  The \emph{combined degree} of the set is the maximum degree (taken with respect to each variable) of the product of all the polynomials in the set:
$
\textstyle \max_{i} ( \deg_{x_i}( \, \prod_j p_j \, )).
$

The set has the \emph{(m,d)-property} if it may be partitioned into $m$ subsets, each with maximum combined degree $d$.
\end{definition}
For example,  
$
\{y^2-x, y^2+1\}
$
has combined degree $4$ and thus the $(1,4)$-property, but also the $(2,2)$-property.

This property (introduced in McCallum's thesis) can measure growth in the projection phase.  In \cite{BDEMW15} we proved that if $A$ has the $(m,d)$-property then $P(A) \cup \cont(A)$ has the $(M,2d^2)$-property with $M = \left \lfloor \tfrac{1}{2}(m+1)^2 \right\rfloor$. When $m>1$, we can bound $M$ by $m^2$ (but we need $2m^2$ to cover $m=1$).

If $A$ has the $(m,d)$-property then so does its squarefree basis.  Hence applying this result recursively (as in Table \ref{tab:P}) measures the growth in $(m,d)$-property during projection under operator (\ref{eq:P}).  After the first projection there are multiple polynomials and so the tighter bound for $M$ is used. 

The number of real roots in a set with the $(m,d)$-property is at most $md$.
The number of cells in the CAD of $\mathbb{R}^1$ is thus bounded by twice the product of the final two entries, plus 1.  Similarly, the total number of cells in the CAD of $\mathbb{R}^n$ by
\begin{align}
\label{eq:BoundP}
(2md+1)\textstyle \prod_{r=1}^{n-1} 
\left[ 2 \left( 2^{2^{r-1}}m^{2^r} \right)\left( 2^{2^{r}-1}d^{2^r} \right) + 1 \right].
\end{align}
Omitting the $+1$s will leave us with the dominant term of the bound, which evaluates to give the following result.

\begin{theorem}
\label{thm:Complexity1}
The dominant term in the bound on the number of CAD cells in $\mathbb{R}^n$ produced using (\ref{eq:P}) is
\begin{align}
&\qquad (2d)^{2^{n}-1}m^{2^{n}-1}2^{2^{n-1}-1}.
\label{eq:DominantTerm:P}
\end{align}
\end{theorem}


From now on assume $\ell$ ECs, $0< \ell \leq \min(m,n)$, all with different mvar.  For simplicity we assume these variables are $x_n, \dots, x_{n-\ell+1}$ (the first $\ell$ projections are reduced).
\begin{lemma}
\label{lem:Complexity}
Suppose $A$ is a set with the $(m,d)$-property and $E \subset A$ has the $(1,d)$-property.  Then $\cont(A) \cup P_{E}^*(A)$ has the $(2m,2d^2)$-property.
\end{lemma}
\begin{proof}
In \cite{BDEMW15} we proved that applying $P_E(A) \cup \cont(A)$ gives a set of $\left \lfloor \tfrac{1}{2}(3m+1) \right\rfloor$ polynomials of combined degree $2d^2$.  
The extra $m-1$ discriminants required by operator (\ref{eq:ECProjStar}) will each have degree at most $d(d-1)$, so pairing them we have $\lceil \tfrac{1}{2}(m-1)\rceil$ sets of combined degree at most $2d^2$.  Then
\[
\left \lfloor \tfrac{1}{2}(3m+1) \right\rfloor + \lceil \tfrac{1}{2}(m-1)\rceil
= m + \left \lfloor \tfrac{1}{2}(m+1) \right\rfloor + \lfloor \tfrac{m}{2}\rfloor
\]
and since $m \in \mathbb{Z}$ this always equals $2m$.
\end{proof}
We apply this recursively in the top half of Table \ref{tab:P2}, with the bottom derived via the process for $P$, as in Table \ref{tab:P}.  

Define $d_i$ and $m_i$ as the entries in the Number and Degree columns of Table \ref{tab:P2} from the row with $i$ Variables. 
We can bound the number of real roots of projection polynomials in $i$ variables by $m_id_i$.  If we lifted with respect to all these projection polynomials, the cell count would be bounded by
\begin{align}
&\textstyle\prod_{i=1}^n \left[2m_id_i + 1\right] 
= \prod_{s=0}^{\ell} \left[ 2 \left( 2^{s}m 2^{2^{s}-1}d^{2^{s}} \right) + 1 \right] 
\nonumber \\
&\qquad \cdot \textstyle \prod_{r=1}^{n-\ell-1} \left[ 
2 \left(  2^{2^{r}\ell}m^{2^{r}} 2^{2^{\ell+r}-1}d^{2^{\ell+r}} \right) + 1 
\right]. \label{eq:boundA}
\end{align}
Omitting the $+1$ from each product allows us to calculate the dominant term of the bound explicitly as
\begin{align}
(2d)^{2^{n}-1} m^{2^{n-\ell} + \ell-1} 2^{\ell 2^{n-\ell} + \ell(\ell-3)/2}.
\label{eq:DominantTerm:P2}
\end{align}


Now we consider the benefit of improved lifting.  
Start by considering the CAD of $\mathbb{R}^{n-(\ell+1)}$.  There can be no reduced lifting until this point and so the cell count bound is given by the second product in (\ref{eq:boundA}), which we will denote by $\dagger$.
The lift to $\mathbb{R}^{n-\ell}$ will involve stack generation over all cells, but only with respect to the EC.  
This can have at most $d_{n-\ell}$ real roots and so the CAD at most $[2d_{n-\ell}+1](\dagger)$ cells.

The next lift, to $\mathbb{R}^{n-\ell-1}$, will lift the sections with respect to the EC, and the  sectors only trivially (to produce the same number of cylinders).  Hence the cell count bound is
$
[2d_{n-(\ell-1)}+1]d_{n-\ell}(\dagger) + (d_{n-\ell}+1)(\dagger)
$
with dominant term $2d_{n-(\ell-1)}d_{n-\ell}(\dagger)$.  Subsequent lifts follow the same pattern and so $2d_nd_{n-1} \dots d_{n-(\ell-1)}d_{n-\ell}(\dagger)$ is the dominant term in the bound for $\mathbb{R}^n$.  This evaluates to give the following result.

\begin{theorem}
\label{thm:Complexity2}
Consider the CAD of $\mathbb{R}^n$ produced using Algorithm \ref{alg:ECM} in the presence of ECs in the top $\ell$ variables.  The dominant term in the bound on the number of cells is
\begin{align}
&\quad 2 \textstyle \prod_{s=0}^{\ell} \left[ 2^{2^{s}-1}d^{2^{s}}  \right] 
\prod_{r=1}^{n-\ell-1} \left[ 
2 \left(  2^{2^{r}\ell}m^{2^{r}} 2^{2^{\ell+r}-1}d^{2^{\ell+r}} \right)  
\right] \nonumber \\
&= (2d)^{2^{n}-1} m^{2^{n-\ell} -2} 2^{\ell 2^{n-\ell} -3\ell}.
\label{eq:DominantTerm:P3}
\end{align}
\end{theorem}

The bound in Theorem \ref{thm:Complexity2} is strictly less than the one in Theorem \ref{thm:Complexity1}. The double exponent of $m$ has decreased by the number of ECs; the result of the improved projection in (\ref{eq:DominantTerm:P2}).  Improved lifting reduced the single exponents further still.

\begin{table}[t]
\vskip-10pt
\caption{Projection under operator (\ref{eq:P}).}  \label{tab:P}
\begin{center}
\begin{tabular}{cccc}
Variables & Number  & Degree                  \\
\hline
$n$       & $m$       & $d$                   \\
$n-1$     & $2m^2$    & $2d^2$                \\
$n-2$     & $4m^4$    & $8d^4$                \\
\vdots    & \vdots    & \vdots                \\
$n-r$     & $2^{2^{r-1}}m^{2^{r}}$   & $2^{2^r-1}d^{2^r}$ \\
\vdots    & \vdots                   & \vdots \\
1         & $2^{2^{n-2}}m^{2^{n-1}}$ & $2^{2^{n-1}-1}d^{2^{n-1}}$
\end{tabular}
\end{center}
\end{table}

\begin{table}[t]
\vskip-20pt
\caption{Projection with (\ref{eq:ECProjStar}) $\ell$ times and then (\ref{eq:P}).}  \label{tab:P2}
\begin{center}
\begin{tabular}{cccc}
Variables    & Number    & Degree                   \\
\hline
$n$          & $m$       & $d$                      \\
$n-1$        & $2m$      & $2d^2$                   \\
\vdots       & \vdots    & \vdots                   \\
$n-\ell$     & $2^{\ell}m$    & $2^{2^{\ell}-1}d^{2^{\ell}}$     \\
$n-(\ell+1)$ & $2^{2\ell}m^2$ & $2^{2^{\ell+1}-1}d^{2^{\ell+1}}$ \\
\vdots       & \vdots                      & \vdots                           \\
$n-(\ell+r)$ & $2^{2^{r}\ell}m^{2^{r}}$    & $2^{2^{\ell+r}-1}d^{2^{\ell+r}}$ \\
\vdots       & \vdots                                   & \vdots              \\
1            & $2^{2^{(n-1-\ell)}\ell}m^{2^{n-1-\ell}}$ & $2^{2^{n-1}-1}d^{2^{n-1}}$ 
\end{tabular}
\end{center}
\vskip-20pt
\end{table}

\section{Conclusions and future work}
\label{SEC:Conclusions}

We have explained how the existing theory for CAD projection using ECs can also be leveraged for significant savings in the lifting phase.  We can reduce both the projection polynomials used for lifting and the cells over which stacks are generated.  We have formalised these ideas in Algorithm \ref{alg:ECM}, verified their use in Theorem \ref{thm:Alg}, and demonstrated the benefit with a worked example and complexity analysis.

A key question is how to best deal with non-primitive ECs?  Consider 
$
\phi := zy = 0 \land \varphi.
$
Under ordering $\dots \succ z \succ y \succ \dots$ the EC $zy=0$ is not primitive, so Algorithm \ref{alg:ECM} cannot use it.  We may be tempted to take $E = \{z\}$ as the primitive part, project with operator (\ref{eq:ECProj}) and include the content $y$ in the first projection.  The CAD of $(y, \dots)$-space would be sign-invariant for $y$ and thus the CAD of $(z,y,\dots)$-space truth invariant for the EC (over admissible cells).  But we can no longer say only sections are admissible for the next lift as there may be cells with $z \neq 0$ and $y=0$.  We could instead lift over all cells.  Alternatively we might rewrite $\phi$ as
$
\phi := (z=0 \land \varphi) \lor (y=0 \land \varphi),
$
so each clause has its own EC. 
The theory of truth-table invariant CADs \cite{BDEMW13, BDEMW15} is designed to deal with such input, but would require its own extension to use beyond the first projection.  Of course, this extension would also be valuable in its own right.

\subsection*{Acknowledgements}

Thanks to the the referees for their helpful comments.
This work was supported by EPSRC grant: EP/J003247/1.

\bibliographystyle{plain}
\bibliography{CAD}

\end{document}